\documentclass{llncs}
\usepackage{amssymb,amsmath}
\usepackage{graphicx,color}

\newcommand\G{\mathcal{G}}

\begin{document}

\title{Cost Sharing in the Aspnes Inoculation Model}

\author{Michael J. Collins}
\institute{Sandia National Laboratories\footnote{Sandia National Laboratories is a multi-program laboratory managed and operated by Sandia Corporation, a wholly owned subsidiary of Lockheed Martin Corporation, for the U.S. Department of Energy's National Nuclear Security Administration under contract DE-AC04-94AL85000.}\\Albuquerque, NM USA 87185\\ \email{mjcolli@sandia.gov}}

\maketitle

\begin{abstract}
We consider the use of cost sharing in the Aspnes model \cite{Aspnes} of network inoculation, showing that this can improve the cost of the optimal equilibrium by a factor of $O(\sqrt{n})$ in a network of $n$ nodes.
\end{abstract}

\section{Introduction}
A game-theoretic model of inoculation against viral infection was introduced by Aspnes et al. \cite{Aspnes}. Each player in this game is a node in an undirected connected graph $\G$ with $n$ nodes. Nodes represent network hosts that might become infected, while edges represent direct communication links through which a virus might spread. Each node has two possible pure strategies: either do nothing, or inoculate itself (i.e. install anti-viral protection). After the nodes have made their choices, an attacker selects one node uniformly at random to infect. Infection then propagates through the graph; a non-inoculated node becomes infected if any of its neighbors are infected.  

Let $I$ be the set of inoculated nodes. These nodes are in effect removed from the graph, leaving a (possibly disconnected) graph of \emph{vulnerable} nodes $\G_I$. Let $v$ be the initial node selected by the attacker. If $v \in I$, then $v$ is not infected and thus no nodes become infected. Otherwise, the infected nodes are precisely the nodes in the connected component containing $v$. Given any node $u$, the connected component containing $u$ is denoted $\kappa_u$, and the size of this component is denoted $k_u$ (with $k_u = 0$ if $u$ is inoculated). Henceforth ``component'' will always mean a connected component of $\G_I$.

Let $C$ be the cost of inoculation, and $L$ be the loss suffered by an infected node. Thus if $u \in I$ , the cost borne by $u$ is just $C$: otherwise $u$'s cost is $L\frac{k_u}{n}$, because $\frac{k_u}{n}$ is the probability of component $\kappa_u$ receiving the initial infection. Given these costs, the \emph{threshold} size for a component is $t=n\frac{C}{L}$. The set $I$ gives a Nash equilibrium if and only if both of these conditions hold:
\begin{itemize}
\item each component of $\G_I$ has size at most $t$
\item de-inoculating any node $j \in I$ (i.e. putting $j$ and all its adjacent edges back in the graph) creates a component of size at least $t$
\end{itemize}
In other words, if a component is larger than $t$, then each node in that component would be better off paying the cost $C$ to inoculate; if an inoculated node could rejoin the graph and still be in a component smaller than $t$, it would be better off not paying the cost of inoculation. 

The \emph{social cost} of $I$ is the sum of individual costs. This sum can be written as
\begin{equation}\label{eq:DetSocialCost}
C|I| + \frac{L}{n}\sum_i k_i^2
\end{equation}

\section{Cost Sharing}
A crucial feature of the Aspnes model is that one node can benefit from another node's decision to inoculate.  This motivates consideration of a model in which one node might pay part of the cost of another node's inoculation. Such \emph{cost-sharing} models of network games have been studied by many authors (see \cite{Charikar08} and references therein), but this idea has not, so far as we are aware, been applied to the Aspnes model previously. We have the same graph as in the original game, but now a pure strategy for node $i$  is a vector $a^i=(a^i_1 \dots a^i_n)$, where $a^i_j$ is the contribution made by node $i$ to the inoculation of node $j$. Node $j$ will be inoculated if and only if
\[
\sum_{1 \leq i \leq n} a_j^i \geq C \ .
\]
The individual cost for node $i$ is 
\[
\sum_{1 \leq j \leq n} a^i_j + L\frac{k_i}{n}
\]
We have the following necessary conditions for equilibria in this game.

\begin{theorem}\label{Th:costshar}
Let $\sigma = (a^1, a^2, \dots a^n)$ be a pure-strategy equilibrium in the cost-sharing Aspnes game. Then
\begin{enumerate}
  \item $\sum_j a^i_j \leq C$ for all nodes $i$ \label{maxcost}
  \item $\sum_i a^i_j$ is either $0$ or $C$ for all nodes $j$ \label{nowaste}
  \item Each $k_i \leq n\frac{C}{L}$ \label{samethresh}
  \item If $j$ is inoculated but de-inoculating $j$ would not increase the size of $\kappa_i$, then $a^i_j = 0$ \label{local}
\end{enumerate}
Also, the cost-sharing game and the original game have the same minimum social cost.
\end{theorem}
\begin{proof}
Any node $i$ violating~(\ref{maxcost}) could reduce its cost by  inoculating itself (i.e. setting $a^i_i=C$) and paying nothing for any other node. If~(\ref{nowaste}) does not hold for some $j$, then any node $i$ with $a^i_j \neq 0$ could reduce its cost by reducing its contribution to $j$ without changing $j$'s inoculation status. If (\ref{samethresh}) does not hold then any node in $\kappa_i$ could reduce its cost by inoculating itself.  Any node $i$ violating~(\ref{local}) could reduce its cost by setting $a^i_j = 0$, since this would not increase $i$'s probability of infection; we call this condition \emph{locality}. Note in particular this implies that, if $i$ is inoculated, then $i$ does not contribute to the cost of inoculating any node other than (perhaps) itself.

Since the cost-sharing game has a strictly larger set of strategies, its minimum social cost can be no greater. On the other hand, let $\sigma = (a^1, a^2, \dots a^n)$ be a strategy vector minimizing social cost in the cost-sharing game (note $\sigma$ need not be an equilibrium), and let $I$ be the resulting set of inoculated nodes. Since cost is minimized, the total amount spent on inoculations is exactly $C|I|$. Thus we can obtain the same social cost in the original game by inoculating all the nodes in $I$ and no others. \qed
\end{proof}

Although cost-sharing cannot improve the social optimum, it can create better equilibria. 
Consider a graph $\G$ with a cost-sharing equilibrium that inoculates a set of vertices $I$, and let $u \in I$. Suppose there are $k$ components $\tau_1, \dots \tau_k$ of $\G_I$ which are adjacent to $u$, with sizes $t_1, \dots t_k$. Let $t=\sum_i t_i$ and let $\hat t_j = t - t_j$. Since de-inoculating $u$ would create a component of size $t+1$, the contribution of  $u$ to its own inoculation can be no greater than $\frac{L}{n}(t+1)$. Because of locality, only nodes from $\cup_i \tau_i$ can contribute to inoculating $u$; therefore there must be some node in some $\tau_j$ whose contribution is at least
\[
\frac{C - L(t+1)/n}{t} = \frac{C}{t} - \frac{L}{n}\frac{t+1}{t} \ .
\]
However, since we are at equilibrium, this contribution must also be no greater than $\frac{L}{n}(\hat t_j + 1)$.  Combining these inequalities yields
\[
\frac{nC}{L} \leq t(\hat t_j + 2) + 1 \ .
\]
This gives a lower bound on $t$: 

\begin{theorem}\label{Th:componentSize}
Let $\G$ be a graph on $n$ nodes and let $I$ be a set of inoculated nodes which is obtained at an equilibrium of the cost-sharing Aspnes game with inoculation cost $C$ and infection cost $L$. Then any component created by de-inoculating $u \in I$ must have size at least
\begin{equation}\label{eq:MinCompSize}
\sqrt{\frac{nC}{L}}-1\ .
\end{equation}
\end{theorem}

Recall that without cost sharing, any component created by de-inoculating a node must have size at least $\frac{C}{L}n$. We now show that with cost sharing, the bound (\ref{eq:MinCompSize}) can be attained, and this can lead to a $O(\sqrt{n})$ improvement in the cost of the best equilibrium. 

\begin{theorem}\label{Th:main}
There exists a family of graphs on which cost-sharing improves the best Aspnes equilibrium by a factor of $O(
\sqrt{n})$.
\end{theorem}
\begin{proof}
Let $\G_n$ be a cycle on $n$ vertices. 
We normalize to $C=1$ and let $n$ approach infinity with $L$ fixed and $L>C$. 

We first show that, even allowing mixed strategies, the social cost of an equilibrium of the original Aspnes game on $\G_n$ must be $O(n)$. 
This expected social cost is 
\begin{equation}\label{eq:ExpSocialCost}
\sum_i \{a_i + (1-a_i)\frac{L}{n}S_i\}
\end{equation}
where $a_i$ is the probability that node $i$ inoculates and $S_i$ is the expected size of $i$'s component if $i$ does not inoculate. At equilibrium, we have $S_i \geq \frac{n}{L}$ when $a_i=1$, $S_i \leq \frac{n}{L}$ when $a_i=0$, and $S_i = \frac{n}{L}$ otherwise \cite{Aspnes}.  A node $i$ is called \emph{nonbinary} if $a_i$ is not zero or one. Note that if a nonbinary node $i$ is adjacent to node $i+1$ with $a_{i+1}=0$ then we must have $S_i=S_{i+1}=\frac{n}{L}$, since the two nodes must be in the same component if $i$ does not inoculate. Similarly, if two adjacent nodes have $a_i = a_{i+1} = 0$ then $S_i=S_{i+1}$.

Now consider a node $j$ with $a_j=1$. If there are no such nodes (or only one such node), then the observations we have just made imply $S_i = \frac{n}{L}$ for all $i$ (or all $i \neq j$), giving an expected social cost of $O(n)$. Otherwise consider the nodes $j+1$ and $j-1$ adjacent to $j$. Note that
\[
S_j = 1 + \exp(k_{j-1}) + \exp(k_{j+1}) 
\]
since if $j$ does not inoculate it merges the components $\kappa_{j-1}$ and $\kappa_{j+1}$. But clearly $S_t \geq \exp(k_t)$ for any node $t$; and since equilibrium requires $S_j \geq \frac{n}{L}$ we must have 
\[S_{j-1} + S_{j+1} \geq \frac{n}{L}-1\ .\]
So suppose without loss of generality that $S_{j+1} \geq \frac{n-L}{2L}$. But this implies we have at least $\frac{n-L}{2L}$ consecutive nodes $j+1, j+2, \dots j+\frac{n-L}{2L}$ with probability of inoculation less than 1: otherwise $S_{j+1}$ could not be this large. Their total contribution to the social cost (\ref{eq:ExpSocialCost}) is $O(n)$, since each such node $j+1 \leq t \leq j+\frac{n-L}{2L}$ has $S_t = S_{j+1}$.

However, with $\sqrt{nL}$ inoculated nodes evenly spaced around the cycle, breaking $\G_n$ into components of size $\sqrt{n/L}-1$, we would have a social cost of
\begin{equation}\label{eq:OptSocialCost}
\sqrt{nL} + \frac{L}{n}\sqrt{nL}\frac{n}{L} = 2\sqrt{nL} \ .
\end{equation}
This set of inoculated nodes (and thus this social cost) can be attained at a pure-strategy equilibrium in the cost-sharing game. If each node shares equally in the cost of protecting the nearest inoculated node, its cost is 
$$\frac{1}{\sqrt{n/L}} + L\frac{\sqrt{n/L}-1}{n} = 2\sqrt{\frac{L}{n}} - \frac{L}{n}$$
By not making any contribution to the cost of inoculation, a node would create a component of size $2\sqrt{n/L}-1$, incurring a cost of 
$$L\frac{2\sqrt{n/L}-1}{n} = 2\sqrt{\frac{L}{n}} - \frac{L}{n}\ .$$
It is easy to see (using the same reasoning as in theorem \ref{Th:costshar}) that no other change in a node's strategy could produce a better result.
\end{proof}

 Note that the fragility of this equilibrium can be removed (while maintaining a ratio of $O(\sqrt{n})$)  by inoculating $(1-\varepsilon)\sqrt{nL}$ evenly-spaced nodes; then a node's cost will strictly increase if it does not contribute to inoculation. Finally, it is easy to see that (\ref{eq:OptSocialCost}) is is in fact the social optimum; given the size of $I$, we minimize $\sum k_i^2$ by distributing the inoculated nodes evenly, so (\ref{eq:DetSocialCost}) becomes
\[
|I| + \frac{L}{n}|I|\left(\frac{n}{|I|}\right)^2
\]
which is minimized at $|I|=\sqrt{Ln}$. Thus the price of stability (defined as the ratio of the social optimum to the best equilibrium) is $1$.

\bibliographystyle{plain}
\bibliography{SANDExample.bib}

\end{document}